\documentclass[journal]{IEEEtran}
\usepackage{cite}
\usepackage{subfig}
\usepackage{amsmath,amssymb,amsfonts,amsthm}
\usepackage{algorithmic}
\usepackage{graphicx}
\usepackage{graphics}
\usepackage{hyperref}
\hypersetup{colorlinks=true,urlcolor=red}
\usepackage{bm} 
\usepackage{xspace}

\newcommand{\xmath}[1] {\ensuremath{#1}\xspace}
\newcommand{\blmath}[1] {\xmath{\bm{#1}}}

\newcommand{\Kb}{{\blmath K}}

\newcommand{\Xb}{{\blmath X}}
\newcommand{\Yb}{{\blmath Y}}

\newcommand{\mb}{{\blmath m}}
\newcommand{\nb}{{\blmath n}}

\newcommand{\wb}{{\blmath w}}
\newcommand{\xb}{{\blmath x}}
\newcommand{\yb}{{\blmath y}}
\newcommand{\zb}{{\blmath z}}

\newcommand{\Hc}{\mathcal{H}}

\newcommand{\Xc}{\mathcal{X}}
\newcommand{\Yc}{\mathcal{Y}}

\newcommand{\Pc}{{{\mathcal P}}}

\newcommand{\Kd}{\mathbb{K}}

\newcommand{\beq}{\begin{equation}}
\newcommand{\eeq}{\end{equation}}
\newcommand{\beqa}{\begin{eqnarray}}
\newcommand{\eeqa}{\end{eqnarray}}
\newcommand{\Fc}{{\mathcal F}}

\newtheorem{theorem}{Theorem}

\newcommand{\Sc}{{{\mathcal S}}}

\newcommand{\Dd}{\mathbb{D}}

\begin{document}

\title{Unsupervised Deep Learning for MR Angiography with Flexible Temporal Resolution}
\date{\vspace{-4ex}}

\author{Eunju~Cha, 
        Hyungjin Chung, 
        Eung Yeop Kim,
        and~Jong~Chul~Ye,~\IEEEmembership{Fellow,~IEEE}
\thanks{E, Cha, H. Chung, and J. C. Ye are with the Department of Bio and Brain Engineering, 
		Korea Advanced Institute of Science and Technology (KAIST), 
		Daejeon 34141, Republic of Korea. 
		J.C. Ye is also with the Department of Mathematical Sciences, KAIST. 
		E.Y.Kim is with Dept. of Radiology, Gachon University College of Medicine, Incheon, Namdong-gu, 21565, Republic of Korea} 
\thanks{This work is supported by National Research Foundation (NRF) of Korea, Grant
		number NRF2016R1A2B3008104.}
}


\maketitle

\begin{abstract}
 Time-resolved  MR angiography (tMRA) has been widely used for dynamic contrast enhanced MRI (DCE-MRI) due to its highly accelerated acquisition. In tMRA,  the periphery of the $k$-space data are sparsely sampled so that  neighbouring frames can be merged to construct one temporal frame. However, this view-sharing scheme fundamentally limits the temporal resolution,
 and it is not possible to change the view-sharing number to achieve different spatio-temporal resolution trade-off. 
Although many deep learning approaches have been recently proposed for MR reconstruction from sparse samples, the existing
approaches usually require
 matched fully sampled $k$-space reference data  for supervised training, which is not suitable for tMRA. This is because
high spatio-temporal resolution ground-truth images are not available for tMRA.
To address this problem,
here we propose a novel unsupervised deep learning  using optimal transport driven cycle-consistent generative adversarial network (cycleGAN). In contrast to the conventional cycleGAN with two pairs of generator and discriminator, the new architecture
requires just a single pair of generator and discriminator, which makes the training much simpler and improves the performance.
Reconstruction results using in vivo tMRA data set confirm that the proposed method can immediately generate high quality reconstruction results at   various choices of view-sharing numbers, allowing us to exploit better trade-off between spatial and temporal resolution  in time-resolved MR angiography.
\end{abstract}

\begin{IEEEkeywords}
Time-resolved MRA, dynamic contrast enhanced MRI,  unsupervised learning, cycleGAN, penalized least squares (PLS), optimal transport
\end{IEEEkeywords}

\IEEEpeerreviewmaketitle

\section{Introduction}\label{sec:introduction}
\IEEEPARstart{D}{CE-MRI} 
 is one of the essential imaging methods for clinical diagnosis. DCE-MRI gives information on the physiological characteristics of tissues, such as blood vessel function, etc., so it is useful for the imaging of strokes or cancers \cite{turnbull2009dynamic,yankeelov2009dynamic}.

In DCE-MRI, after injecting the contrast agent to patient, a sequence of MR images is obtained.
Since the motion during the acquisition of $k$-space data can degrade the quality of MR images, 
multiple studies have been conducted to accelerate DCE-MRI with improved temporal resolution.
Specifically, time-resolved MRA such as time-resolved angiography with interleaved stochastic trajectories (TWIST) \cite{laub2006syngo} has been widely used as  one of the solutions for enhancing the temporal resolution of DCE-MRI.  Here, 
the periphery of $k$-space data is very sparsely sampled at each time frame while the center of $k$-space frequency data is fully acquired for the retention of image contrast. 
Then, the high frequency regions of $k$-space data from several time frames are combined together to form a single $k$-space data while the low frequency of $k$-space data remains at intact.  This view-sharing leads to a uniformly subsampled $k$-space sampling pattern as shown in Fig.~\ref{fig:TWIST_concept}(a), after which
the generalized autocalibrating partial parallel acquisition (GRAPPA) \cite{griswold2002generalized} can be applied to reconstruct images from the uniformly downsampled $k$-space data. 

\begin{figure}[!t] 	
\center{ 
\includegraphics[width=8cm]{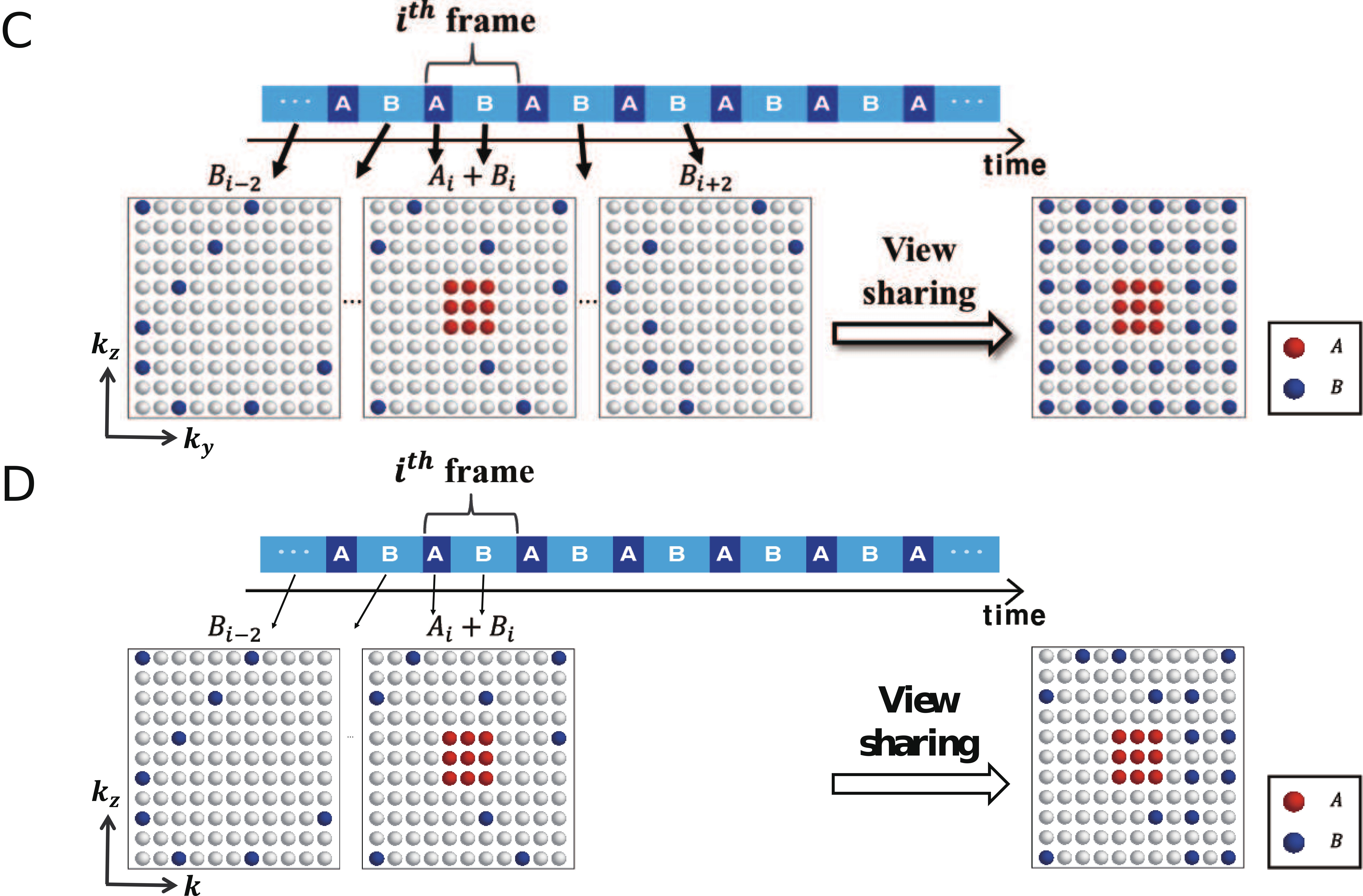}
\caption{TWIST sampling scheme used in this paper. The center and periphery of $k$-space are sampled at A and B frames, respectively. (a) Conventional sampling scheme for 2D GRAPPA reconstruction, and  (b) an example of  reduced view sharing for our method.}
\label{fig:TWIST_concept}
}
\end{figure}

Although this results in a highly accelerated acquisition with a noticeable enhancement in both temporal and spatial resolution, 
 the view sharing from several temporal frames results in the blurring of the temporal resolution. 
 In fact, the resulting inaccuracy in the
temporal resolution is considered as the main huddle when using tMRA as a quantitative tool for perfusion study \cite{lim20083d, kim2019advanced}.
  Unfortunately, the current sampling scheme with fixed temporal window is not flexible for the
reconstruction at reduced sliding window sizes to improve temporal resolution, 
 since $k$-space data of the reduced view sharing is a subset of uniformly sampled $k$-space data, which results in the coherent
aliasing artifacts that are difficult to overcome by parallel imaging or compressed sensing.
 Therefore, the current methods have limitation in investigating various spatio-temporal resolution trade-off to accurately
 quantify the perfusion dynamics.
%
 
%
 To overcome this limitation, a
$k$-space interpolation algorithm \cite{cha2017true} using annihilating filter-based low rank Hankel matrix approach (ALOHA) \cite{jin2016general, lee2016acceleration, lee2016reference} was proposed to synergistically combine CS-MRI and parallel MRI (pMRI) as $k$-space interpolation problems. Nonetheless, high computational complexity for matrix decomposition, which is essential in this algorithm, is an obstacle to practical applications.

Recently, deep learning approaches  have been extensively studied for accelerated
MRI \cite{hammernik2018learning, han2017deep, han2019k, kwon2017parallel, lee2018deep, schlemper2018deep, wang2016accelerating}.
In particular, in \cite{hammernik2018learning}, deep network architecture using unfolded iterative compressed sensing (CS) algorithm was proposed. Domain adaptation \cite{han2017deep}, deep  residual learning \cite{lee2018deep}, and data consistency layers \cite{schlemper2018deep} were properly utilized to reconstruct CS-MRI. All of these methods have demonstrated dramatic performance improvement over the CS approaches \cite{lustig2007sparse,jung2009k,lingala2011accelerated,jin2016general,lee2016acceleration}, while significantly reducing computational complexity.
Unfortunately,  most of the existing deep learning approaches require matched ground-truth data for supervised training,
which is not applicable to the tMRA since high spatial resolution with fast temporal resolution ground-truth images
cannot be acquired in practice. Therefore, unsupervised neural network training without matched reference data
is required.

To address this issue, 
here we propose a novel   cycleGAN architecture inspired by
the optimal transport (OT) theory \cite{villani2008optimal, peyre2019computational,lim2019cyclegan, sim2019optimal}.
One of the important advantages of our cycleGAN is that
the necessity of obtaining fast-temporal resolution data is no more required for training,  as long
as we can obtain {\em unmatched} reference data of high spatial resolution.
Moreover,   once the network is trained, subsampled $k$-space data at various view sharing can be used as input to generate reconstruction results at  various spatial and temporal resolutions. 
In terms of network architecture,  our
cycleGAN architecture is much simpler than the conventional one \cite{zhu2017unpaired},
since
one of the generators in the standard cycleGAN can be replaced with a deterministic operation so that
we only need a single pair of generator and discriminator. This makes training more stable and leads to better reconstruction quality than the original cycleGAN. 
Using experimental results, we verify that our cycleGAN   can reconstruct high quality tMRA  at various temporal resolutions.

\section{Related Works}

\subsection{Optimal Transport Driven CycleGAN}

Here
 we briefly review the OT-driven cycleGAN design in our companion paper \cite{sim2019optimal}, which is used
for our unsupervised learning method.

Consider an inverse problem, where  a noisy  measurement $\yb \in \Yc$ from an unobserved image 
$\xb \in \Xc$ is modeled by
\begin{eqnarray}
\yb&=&\Hc \xb +\wb \ ,
\end{eqnarray}
where $\wb$ is the measurement noise and $\Hc : \Xc \mapsto \Yc$ is the {\em known} measurement 
operator. 
Since inverse problems are ill-posed,
the penalized least squares (PLS) approach is a classical strategy to mitigate the 
ill-posedness:
\begin{eqnarray}\label{eq:st}
\hat \xb = \arg \min_\xb c(\xb;\yb) := \|\yb - \Hc \xb\|^2 + R(\xb)
\end{eqnarray}
 where $R(\xb)$ is a penalty function to impose the constraint to the reconstruction image.

{
 In our companion paper\cite{sim2019optimal}, we proposed a new PLS cost function with a novel deep learning prior:}
\begin{eqnarray}\label{eq:cost0}
c(\xb,\yb;\Theta)=\|\yb-\Hc\xb\|+ \| G_\Theta(\yb) -\xb\|
\end{eqnarray}
 where $G_\Theta$ is a  neural network with the network 
parameter $\Theta$ and input $\yb$.
 The new PLS cost in \eqref{eq:cost0}  was proposed so that
 the unknown $\xb$ is estimated under the constraint that there exists an inverse mapping from the measurement $\yb$ 
by a neural network $G_\Theta$.
The new PLS cost function has many important  advantages over existing PLS in \eqref{eq:st}. In particular,
if  the global minimizer is achieved, i.e. $c(\xb,\yb;\Theta)=0$,  then  we have
\begin{eqnarray*}
 \yb=\Hc \xb,\quad \xb=G_{\Theta}(\yb)
\end{eqnarray*}
Therefore, $G_{\Theta}$ can be an inverse of the forward operator $\Hc$, which is the ultimate goal
in the inverse problem.

Since we do not have any matched reference data,
we further assume that  $\xb,\yb$ in \eqref{eq:cost0} are random vectors
with the measure $\mu$ and $\nu$, respectively,
so that the estimation of the parameter $\Theta$ should be done by considering
all realizations of $\xb,\yb$. 
Therefore, our goal of the unsupervised learning is to find the parameterized map $G_\Theta:\Yc \mapsto \Xc$ such
that average cost with respect to some joint distribution $\pi(x,y)\in P(\Xc\times \Yc)$ can be minimized.
This 
is equal to finding the transportation mapping between two probability measures $\mu\in P(\Xc)$ and  $\nu\in P(\Yc)$ \cite{villani2008optimal, peyre2019computational}.
In particular, rather than choosing arbitrary joint distributions, the optimal transport theory \cite{villani2008optimal, peyre2019computational} informs us that
the optimization problem should be formulated with respect to the optimal transportation cost  
%
\begin{eqnarray}
\label{eq:avg_transport_cost}
\mathbb{K}(\Theta) &:=  \min_{\pi} \ \int_{\Xc \times \Yc} c(\xb, \yb;\Theta) d\pi(\xb, \yb),
\end{eqnarray}
where the minimum is taken over the joint distribution $\pi(x,y)$ whose marginal distribution with respect to $\Xc$ and $\Yc$ is $\mu$ and $\nu$, respectively. 
Another important discovery in our companion paper \cite{sim2019optimal}   is that the resulting primal problem
can be equivalently represented by the Kantorovich dual formulation \cite{villani2008optimal, peyre2019computational}:
\begin{equation*}
\min_{\Theta} \mathbb{K}(\Theta) = \min_{\Theta} \max_{\Upsilon} \ \ell (\Theta; \Upsilon),
\end{equation*}
where 
\begin{equation}
\label{eq:cycle_gan_loss}
\ell(\Theta; \Upsilon) = \gamma\ell_{cycle}(\Theta) + \ell_{WGAN}(\Theta; \Upsilon) \quad .
\end{equation}
Here, $\gamma$ denotes some hyper-parameter, and
$\ell_{cycle}$ and $\ell_{WGAN}$ are simplified as
\begin{eqnarray}
\begin{split}
\ell_{cycle}(\Theta) =  \int_\Xc \|\xb - G_\Theta (\Hc \xb)\|d\mu(\xb) \\
+  \int_\Yc \|\yb-\Hc G_\Theta (\yb)\| d\nu(\yb)
\end{split}
\label{eq:our cycle loss}
\end{eqnarray}
\begin{eqnarray}
\begin{split}
&\ell_{WGAN}(\Theta; \Upsilon) \\
&= \left(\int_\Xc \varphi_\Upsilon(\xb)d\mu(\xb) - \int_\Yc \varphi_\Upsilon(G_\Theta (\yb))d\nu(\yb)\right)
\end{split}
\label{eq:our WGAN loss}
\end{eqnarray}
where $\varphi_{\Upsilon}$  is  the Kantorovich 1-Lipschitz potential.
Here, the term $\ell_{cycle}$ is  the cycle-consistency loss,  $\ell_{WGAN}$ is the
Wasserstein GAN (WGAN) loss \cite{arjovsky2017wasserstein}, and $\varphi_\Upsilon$ is often
called the discriminator.
It is important to note that we just have one discriminator in \eqref{eq:our WGAN loss}, since
$\Hc$ is a known deterministic generator and we only require one single CNN generator $G_\Theta(y)$ \cite{sim2019optimal}.

Interestingly, the resulting dual formulation shows that
the physics-driven data consistency term is used to stabilize the training of  the neural network. 
This implies that in contrast to the inverse formulation using deep learning prior \cite{zhang2017learning,aggarwal2018modl}  that
incorporate the physics-driven information during the run-time reconstruction and network architectures,
our cycleGAN formulation shows that the physics-driven constraint
can be incorporated in training a feed-forward deep neural network so that it generates physically
meaningful estimates \cite{sim2019optimal}.

\begin{figure*}[!t] 	
\centerline{\includegraphics[width=15 cm]{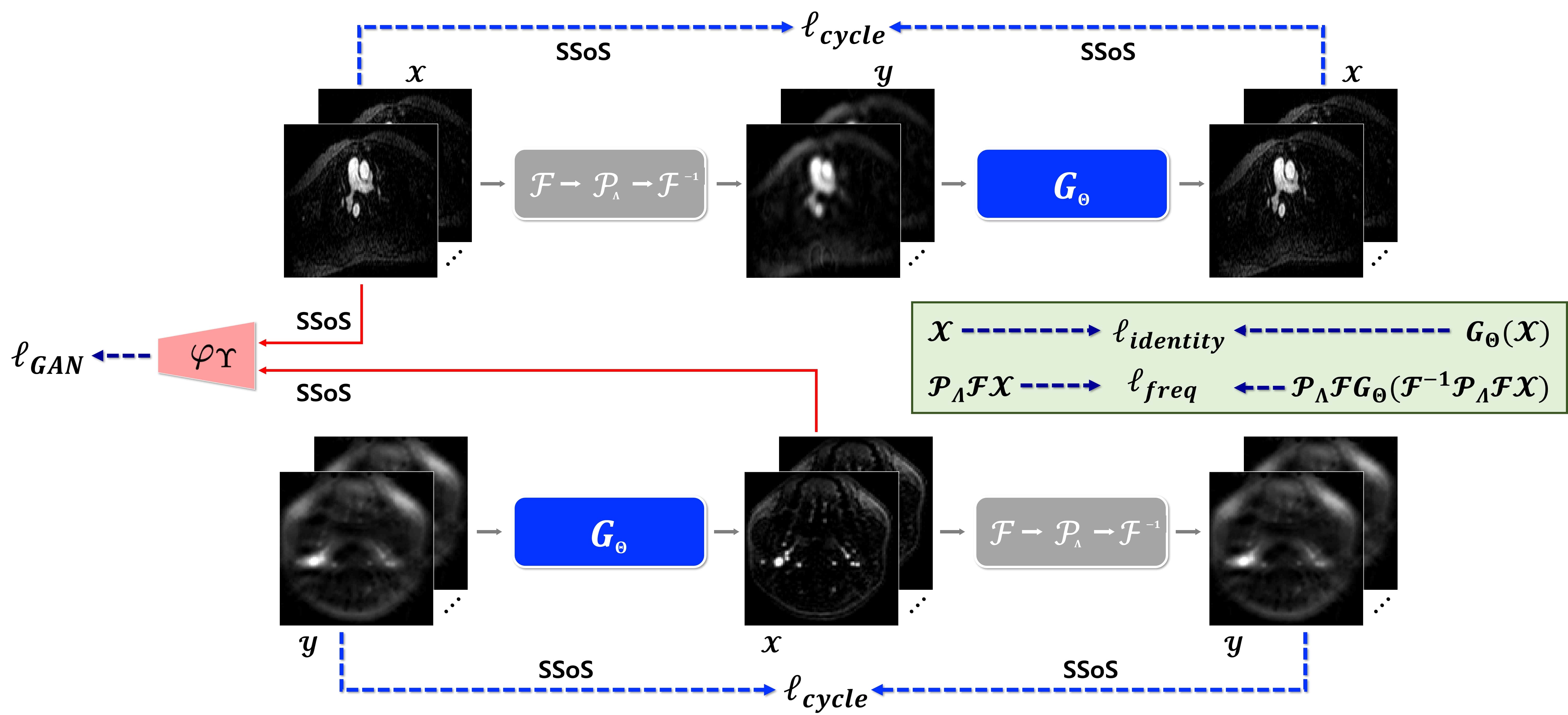}}
\caption{{Overview of the proposed cycleGAN architecture for TWIST imaging with reduced view sharing. There is one pair of generator $G_\Theta$ : $\Yb \to \Xb$ and discriminator $\psi_\Xi$. There are four losses $-$ adversarial loss, cyclic loss, frequency loss, and the identity loss in training the network.}}
\label{fig:overview}
\end{figure*}

\section{Theory}

 \subsection{Problem Formulation }

  To improve the temporal resolution of tMRI,
consider the reduction of the view-sharing number
 as shown in Fig. \ref{fig:TWIST_concept}(b). Specifically, 
for a given under-sampling pattern $\Lambda$ from the reduced view sharing, 
 the forward measurement from $C$-channel receiver coils is given by
\begin{eqnarray}\label{eq:meas}
\widehat \Xb = \Pc_\Lambda \Fc \Xb 
\end{eqnarray}
where
\begin{align*}
 \Xb&=[\xb^{(1)},\cdots, \xb^{(C)}] \\
  \widehat\Xb&=[\hat \xb^{(1)},\cdots, \hat\xb^{(C)}]
 \end{align*}
 where $\xb^{(i)}$ is the $i$-th channel unknown image,
$\Fc$ is the 
2-D Fourier transform, $\Pc_\Lambda$ is the projection to $\Lambda$ that denotes $k$-space 
sampling indices,  and $\hat\xb^{(i)}$ is the corresponding downsampled $k$-space.
In this case,  it is difficult to apply GRAPPA directly due to the irregular under-sampling pattern,
so our goal is to obtain a deep learning approach to address this problem.
However, the main technical difficulty is that it is not possible to obtain full $k$-space
data that can be used as reference for supervised training, since during the full
$k$-space acquisition, the image content $\Xb$ in \eqref{eq:meas} changes due to the
dynamic nature of contrast agent distribution.

\subsection{PLS Transportation Cost for Optimal Transport}

To address this problem, here we employ the unsupervised learning method as an extension of
OT-driven cycleGAN \cite{sim2019optimal}.
For this, we first obtain the aliased reconstruction in the image domain:
\begin{equation}
\label{eq:alias}
\Yb =  \mathcal{F}^{-1} \Pc_{\Lambda} \mathcal{F} \Xb .
\end{equation}
The reason we choose the aliased image domain as $\Yc$ instead of $k$-space measurement is that the construction of the discriminator in the image domain
is well-established.

Using \eqref{eq:meas} and  \eqref{eq:alias},  we extend the PLS transportation cost in \eqref{eq:cost0} to the following formulation:
%
\begin{align}
c(\Xb, \Yb; \Theta) &= d_{I}\left(\Yb,\mathcal{F}^{-1} \Pc_{\Lambda} \mathcal{F} \Xb\right) \label{eq:c1} \\
& + d_{I} \left(\Xb,  G_\Theta(\Yb)\right) \label{eq:c2} \\
&+\alpha d_{I}\left(\Xb, G_\Theta(\Xb) \right) \label{eq:c3} \\
&+ \beta d_{F}\left(\Pc_\Lambda\Fc \Xb, \Pc_\Lambda \Fc G_\Theta(\mathcal{F}^{-1} \Pc_{\Lambda} \mathcal{F}\Xb)\right) \label{eq:c4}
\end{align}
where  $\alpha,\beta>0$ are regularization parameters,
$d_{I}(\cdot,\cdot)$  and $d_{F}(\cdot,\cdot)$  are distance metrics for the image and $k$-space domain, respectively,
whose definition will be discussed later.
More specifically, \eqref{eq:c1} and \eqref{eq:c2} denote  the data fidelity term and deep learning based prior term similar to those in \eqref{eq:cost0}.
The last two terms \eqref{eq:c3} and \eqref{eq:c4}
denote the identity loss for the spatial domain and frequency domain, respectively.
More specifically, \eqref{eq:c3} ensures that  the generator $G_\Theta$ should not alter the fully sampled high resolution image,
and \eqref{eq:c4} enforces that the acquired $k$-space samples should be maintained.

Using \eqref{eq:c1}-\eqref{eq:c4}, the primal form of the optimal transport problem becomes
\begin{eqnarray}
\label{eq:primal}
\mathbb{K}(\Theta) &:=  \min_{\pi} \ \int c(\Xb, \Yb;\Theta) d\pi(\Xb, \Yb),
\end{eqnarray}
Theorem~\ref{thm:dual} then specifies the corresponding 
  Kantorovich dual formulation:
  \begin{theorem}\label{thm:dual}
  For the given primal optimal transport problem in \eqref{eq:primal} with \eqref{eq:c1}-\eqref{eq:c4}, 
  suppose that there exist $\varphi_\Upsilon$ such that 
  \begin{align}\label{eq:bound}
  \|\varphi_\Upsilon(\Xb)-\varphi_\Upsilon(\Xb')\|\leq d_{I}(\Xb,\Xb') ,\quad \forall \Xb,\Xb'
  \end{align}
  then 
  the associated
  Kantorovich dual formulation is given by
\begin{equation*}
\min_{\Theta} \mathbb{K}(\Theta) = \min_{\Theta} \max_{\Upsilon} \ \ell (\Theta; \Upsilon),
\end{equation*}
Here,
\begin{align}
\label{eq:total_cost}
\ell(\Theta; \Upsilon)  &= \gamma\ell_{cycle}(\Theta) + \ell_{WGAN}(\Theta; \Upsilon) \notag \\
&+\alpha \ell_{identity}(\Theta)+\beta\ell_{freq}(\Theta)
\end{align}
where $\gamma$ denotes some hyper-parameter, and
 $\ell_{cycle}$ and $\ell_{WGAN}$ are given by
\begin{align}
\label{eq:simple_cycle_loss}
\ell_{cycle}(\Theta)  &= \int_{\Yc} d_{I}\left(\Yb, \mathcal{F}^{-1} \Pc_{\Lambda} \mathcal{F}G_\Theta(\Yb)\right) d\nu(\Yb) \notag\\
& +    \int_{\Xc} d_{I}\left(\Xb, G_\Theta(\mathcal{F}^{-1} \Pc_{\Lambda} \mathcal{F} \Xb)\right) d\mu(\Xb)  ,
\end{align}
and
\begin{align}
\label{eq:simple_wgan_loss}
\ell_{WGAN}(\Theta,\Upsilon)  &= \int_{\Xc} \varphi_{\Upsilon}(\Xb)d\mu(\Xb) \notag\\
&- \int_{\Yc} \varphi_{\Upsilon}(G_\Theta(\Yb)) d\nu(\Yb) .
\end{align}
Moreover, the image and $k$-space identity loss are given by
\begin{align}\label{eq:identity}
\ell_{identity}(\Theta)  &= \int_{\Xc} d_{I}\left(\Xb, G_\Theta(\Xb) \right)d\mu(\Xb)  
\end{align}
\begin{align}\label{eq:freq}
\ell_{freq}(\Theta)  &= \int_{\Yc} d_{F}\left(\Pc_\Lambda\Fc \Xb, \Pc_\Lambda \Fc G_\Theta(\mathcal{F}^{-1} \Pc_{\Lambda} \mathcal{F}\Xb) \right) d\mu(\Xb)
\end{align}
\end{theorem}
\begin{proof}
See Appendix.
\end{proof}

\subsection{Choice of Metric}

The metric in the image and the $k$-space domain should be defined
according to their inherent properties.
For example, the $k$-space data for each
channel should be measured separately to keep the integrity of the Fourier data.
So the difference are calculated for each channel and added back together to obtain the total error. This leads to the choice
of  $d_{F}(\cdot)=\|\cdot\|_F$, i.e. the Froebenius norm.
On the other hand,
to deal with the multi-channel images in which each channel image is sensitive to specific spatial locations,
we define the metric $d_{I}(\cdot,\cdot)$ as follows:
\begin{eqnarray}\label{eq:Sc}
 d_{I} \left(\Xb,  \Xb'\right) = \|\Sc(\Xb) - \Sc(\Xb)'\|
\end{eqnarray}
%
where 
 $\Sc$  denotes the element-wise square-root of sum of squares (SSoS) operation
such that $\zb=\Sc(\Xb)$ is composed of elements
\begin{align}\label{eq:sos}
z_n =  \left(\sum_{i=1}^C |x^{(i)}_n|^2 \right)^{\frac{1}{2}}
\end{align}
where $x_n^{(i)}$ is the $n$-th elements of the $i$-th coil image.
By measuring the image domain difference in terms of SSoS images,
we found that the algorithm is less prone to each specific coil map.
Furthermore, using the definition \eqref{eq:Sc}, we can easily see that
\begin{eqnarray}
 \|\varphi_\Upsilon(\Sc(\Xb))-\varphi_\Upsilon(\Sc(\Xb'))\|\leq d_{I} \left(\Xb,  \Xb'\right)
\end{eqnarray}
if we choose $\varphi_\Upsilon$ as 1-Lipschitz function with respect to the SSoS image.
The overall flowchart of the proposed cycleGAN architecture is illustrated in Fig. \ref{fig:overview}.

\section{Method}


\subsection{Training dataset}
Twelve sets of in vivo 3D DCE data were acquired with Siemens 3T Verio scanners using TWIST sequence in Gachon University Gil Medical center. The area of scans were fixed to visualize carotid vessels.
  Among twelve sets of given data, there were three different sets of scan parameters. The first eight sets have repetition time (TR) = 2.74 ms, echo time (TE) = 1.02 ms, 241$\times$640$\times$106 matrix, 1.0 mm slice thickness, 16 coils, and 16 temporal frames. The next two sets have TR = 2.5 ms, TE = 0.94 ms, 159$\times$640$\times$80 matrix, 1.2 mm slice thickness, 16 coils, and 30 temporal frames. Finally, for the last two sets, the acquisition parameters were the same as the two sets described beforehand, with the only difference in 2.5 mm slice thickness and 37 temporal frames. Specific sampling patterns are described  in Fig. \ref{fig:TWIST_concept}(a), where 24$\times$24 size ACS region were obtained for the autocalibration of 2D GRAPPA kernel. Furthermore, data were acquired using partial Fourier scheme such that only 63$\%$ of the data was acquired. 
For GRAPPA reconstruction, high frequency $k$-space data of five time frames, ($B_{i-2}, \cdots, B_{i+2}$) are combined with four center sampled
  frames ($A_{i-1}, \cdots, A_{i+1}$), to generate a single $k$-space which is down-sampled by a factor 3 and 2 along $k_y$ and $k_z$ directions, respectively.  Then,  2-D GRAPPA \cite{griswold2002generalized} is used to estimate missing $k$-space data. 
Accordingly, the resulting sliding window size  is 9 frames so that temporal blurring is unavoidable.
  On the other hand, images in the under-sampled domain were generated utilizing reduced view sharing as depicted in Fig. \ref{fig:TWIST_concept}(b). Out of twelve patient data, four sets were used for training, which corresponds to 33,890 slices of training data. The remaining eight sets were used for testing, which corresponds to 44,850 slices of test data.

\begin{figure}[!hbt] 	
\centerline{\includegraphics[width=9cm]{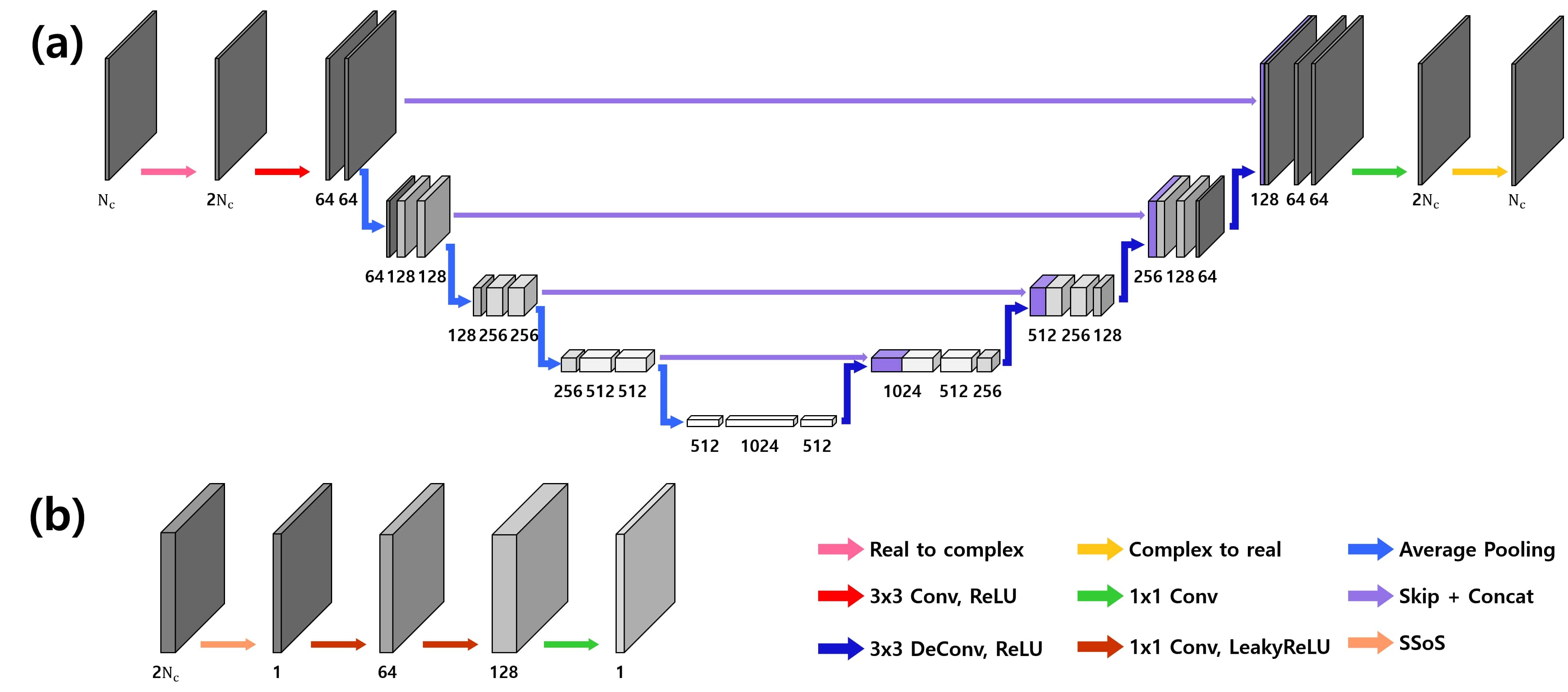}}
\caption{Network architecture of (a) generator $G_\Theta$ and (b) discriminator $\varphi_\Upsilon$.}
\label{fig:network}
\end{figure}

\subsection{Network architecture}

\begin{figure*}[!hbt] 	
\center{ 
\includegraphics[width=18cm]{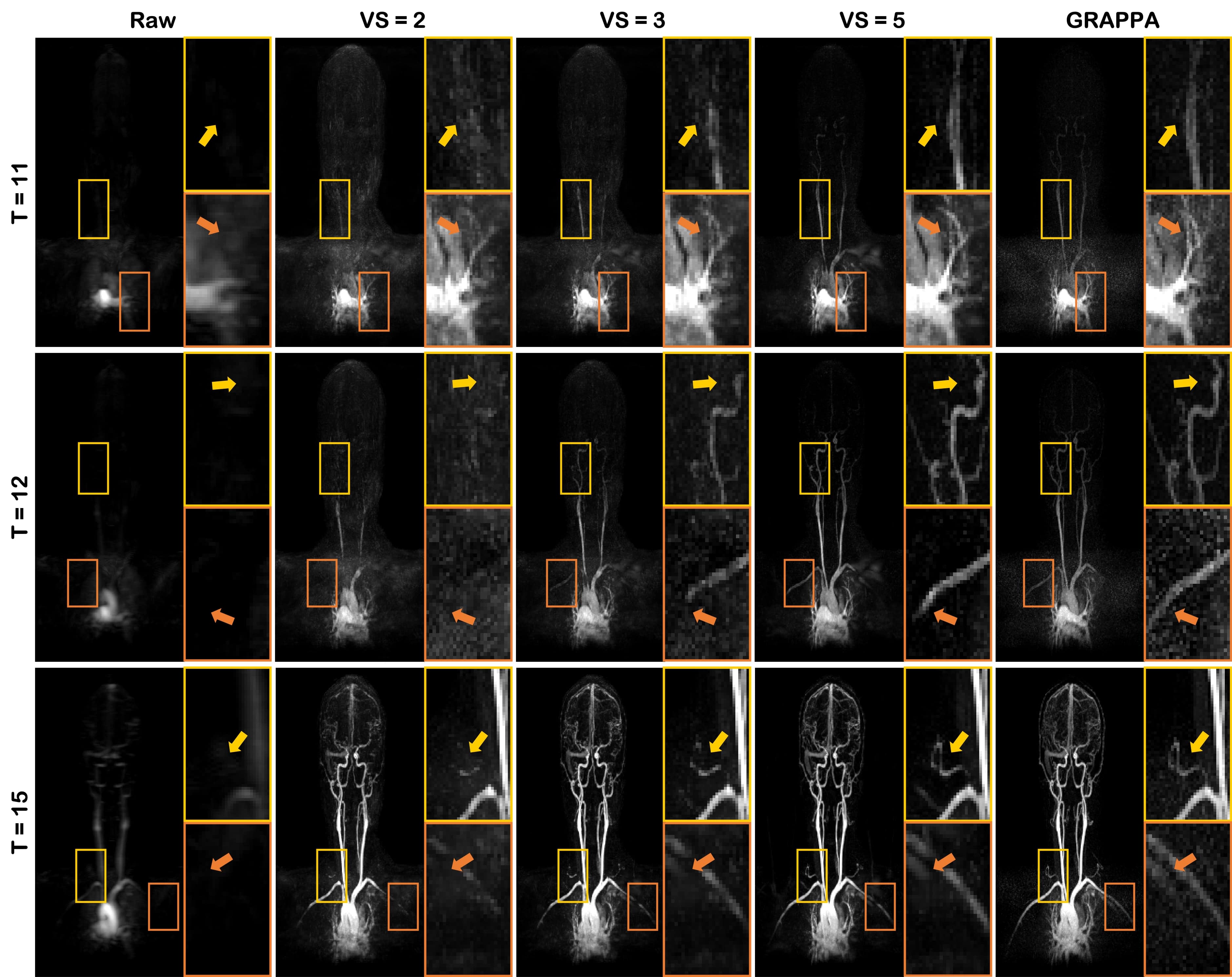}
}
\caption{Temporal resolution comparison of raw data,  and reconstruction results by GRAPPA with VS=5 and the proposed method at various VSs. Here, VS stands for the number of view sharing.}
\label{fig:temporal_resol}
\end{figure*}

Our network consists of a single generator and a single discriminator. The generator $G_\Theta$ adopts standard U-Net architecture \cite{ronneberger2015u}, which has large receptive field and and is widely used in various tasks. More specifically, our network consists of blocks that are comprised of 3$\times$3 convolution, instance normalization \cite{ulyanov2016instance}, and rectified linear unit (ReLU) \cite{nair2010rectified}. Skip connections are also added with concatenation for easier gradient flow. On each stage, three blocks of convolution-instance norm-ReLU are placed except for the last stage which consists of a single $1 \times 1 $ convolution layer (green arrow in Fig. \ref{fig:network}(a)). Detailed depiction of the network can be found in Fig.  \ref{fig:network}(a). As for the pooling, we adopted $2 \times 2$ average pooling. The number of convolutional filters were set to 64 at the first stage, which increased 2-fold at each stage and reached 1024 at the final stage.
 To handle the inherent nature of MR images that are complex,  the real and imaginary part of the complex data are stacked in the channel dimension. Accordingly, the number of input channels was set to 32 (16 coils $\times$2 = 32).

For the discriminator $\varphi_\Upsilon$ illustrated in Fig. \ref{fig:network}(b), we adopt $1 \times 1$ PatchGAN discriminator from \cite{isola2017image}. For stable and efficient training, our discriminator takes a single channel image data constructed by the SSoS operator from complex multi-coil image.
 To illustrate further, our discriminator contains  three $1 \times 1$ convolution layers followed by instance normalization and LeakyReLU. The first convolutional layer consists of 64 sets of $1 \times 1 $ convolution layer, and the number of convolution kernels in the second layer is 128. At the last layer, $1 \times 1$ convolution layers is appended to compute the final feature map.

\subsection{Network training}

The hyper-parameters in   \eqref{eq:total_cost} were set to $\alpha=1.0$, $\beta=2.0$ and $\gamma=2.0$.
Aliased images from sub-sampled $k$-space data with various view sharing factors  (see Figs.~\ref{fig:TWIST_concept}(b)) were used for data in domain $\Yc$, whereas reconstructed multi-coil images using 2-D GRAPPA which utilizes all the view-sharing were used for $\Xc$ since it generates
best spatial resolution images. Note that the two domains do not match each other, since the temporal resolution are different.
The down-sampling mask $\Pc_\Lambda$ was selected randomly from one of the view sharing numbers (VS) = 2, 3, or 5. Accordingly, we could use all three types of sub-sampled data with VS=2, 3, and 5  as inputs 
so that we can compare various levels of spatial- and temporal- resolution to analyze its trade-off.
For GRAPPA reconstruction, parameters were chosen such that optimal results could be obtained. Here, we use kernel size of 5$\times$5 for interpolation.

For optimization, our network was trained with Adam optimizer \cite{kingma2014adam} with momentum $\beta _1 = 0.5$ and $\beta_2 = 0.999$. The network was trained for 50 epochs which consists of two phases. In the first phase, learning rate was specified to 0.001 for 10 epochs. Next, in the second phase learning rate was linearly degraded to 0 for the remaining 40 epochs. At each step, generator was updated 5 times per single update of the discriminator. Our batch size was set to 1. For pre-processing, each individual image was normalized with the standard deviation of the under-sampled image. The proposed network was implemented in Python using PyTorch library \cite{paszke2017automatic} and trained using an NVidia GeForce GTX 1080-Ti graphics processing unit. It took about 6 days for the network training.

\subsection{Comparative Studies}

As for CS reconstruction methods, ALOHA \cite{jin2016general} and k-t SLR \cite{lingala2011accelerated} were chosen for comparison.
Reconstruction parameters for ALOHA were given as follows: annihilating filter size = 13$\times$5 , 3 levels of pyramidal decomposition, decreasing LMaFit tolerance values ($10^{-3}, 10^{-4}, 10^{-5}$) at each level, and ADMM parameter $\mu$ = $10^{-1}$. For k-t SLR, parameters were determined as: the value of p in Schatten p-norm = 0.1, regularization parameter for Schatten p-norm $\mu_1 = 10^{-10}$, and the number of outer iterations = 15.

To quantitatively compare the performance measure of the proposed algorithm as opposed to algorithms presented for comparative study, the peak signal-to-noise ratio (PSNR) and structural similarity (SSIM) index \cite{wang2004image} were computed following the standards. Reconstructed outputs of our network is multi-coil complex-valued data whereas in clinical situations we need a single magnitude image. Therefore, we perform the SSoS to the multi-coil data to obtain a single image for comparison. We calculate the metrics between the images after SSoS operation.
More specifically, PSNR is defined as follows:
\begin{eqnarray}
	PSNR 
		 &=& 20 \cdot \log_{10} \left(\dfrac{MAX_{\zb}}{\sqrt{MSE(\widetilde{\zb}, \zb)}}\right), 
\label{eq:psnr}		 
\end{eqnarray}
where $\widetilde{\zb}$ stands for reconstructed sum-of-squares image and $\zb$ holds for noise-free sum-of-squares image (ground truth). $MAX_{\zb}$ is the maximum pixel intensity of the ground truth.
The SSIM is defined to better capture the perceptual similarity between the original image and the distorted image, and is defined as follows:
\begin{equation}
	SSIM = \dfrac{(2\mu_{\widetilde{\zb}}\mu_{\zb}+c_1)(2\sigma_{\widetilde{\zb}\zb}+c_2)}{(\mu_{\widetilde{\zb}}^2+\mu_{\zb}^2+c_1)(\sigma_{\widetilde{\zb}}^2+\sigma_{\zb}^2+c_2)},
\end{equation}
where $\mu_{\mb}$, $\sigma_{\mb}^2$, and $\sigma_{\mb\nb}$ denote average and variance of $\mb$, and covariance of $\mb$ and $\nb$, respectively. For numeric stability, $c_1=(k_1R)^2$ and $c_2=(k_2R)^2$ are added where $R$ is the dynamic range of pixel values. We keep the default values $k_1 = 0.01$ and $k_2 = 0.03$.

\section{Results}\label{sec:result}

\subsection{In Vivo Results}

Reconstruction results of the carotid vessel data sets from in vivo acquisition are demonstrated in Fig. \ref{fig:temporal_resol}. 
The temporal frames were chosen to illustrate the propagation of the contrast agent and to compare the temporal resolution. In the proposed method, a single neural network trained with various view sharing numbers is sufficient to provide VS-agnostic results. Thus, we provide reconstruction results with all different view sharing - 2, 3, and 5. Raw data shown in Fig. \ref{fig:temporal_resol} were obtained by directly applying inverse Fast Fourier Transform (FFT) to the $k$-space data without view sharing, which reflects the true temporal resolution despite its low visual quality at each time frame.

\begin{figure}[!hbt] 	
\center{ 
\includegraphics[width=8cm]{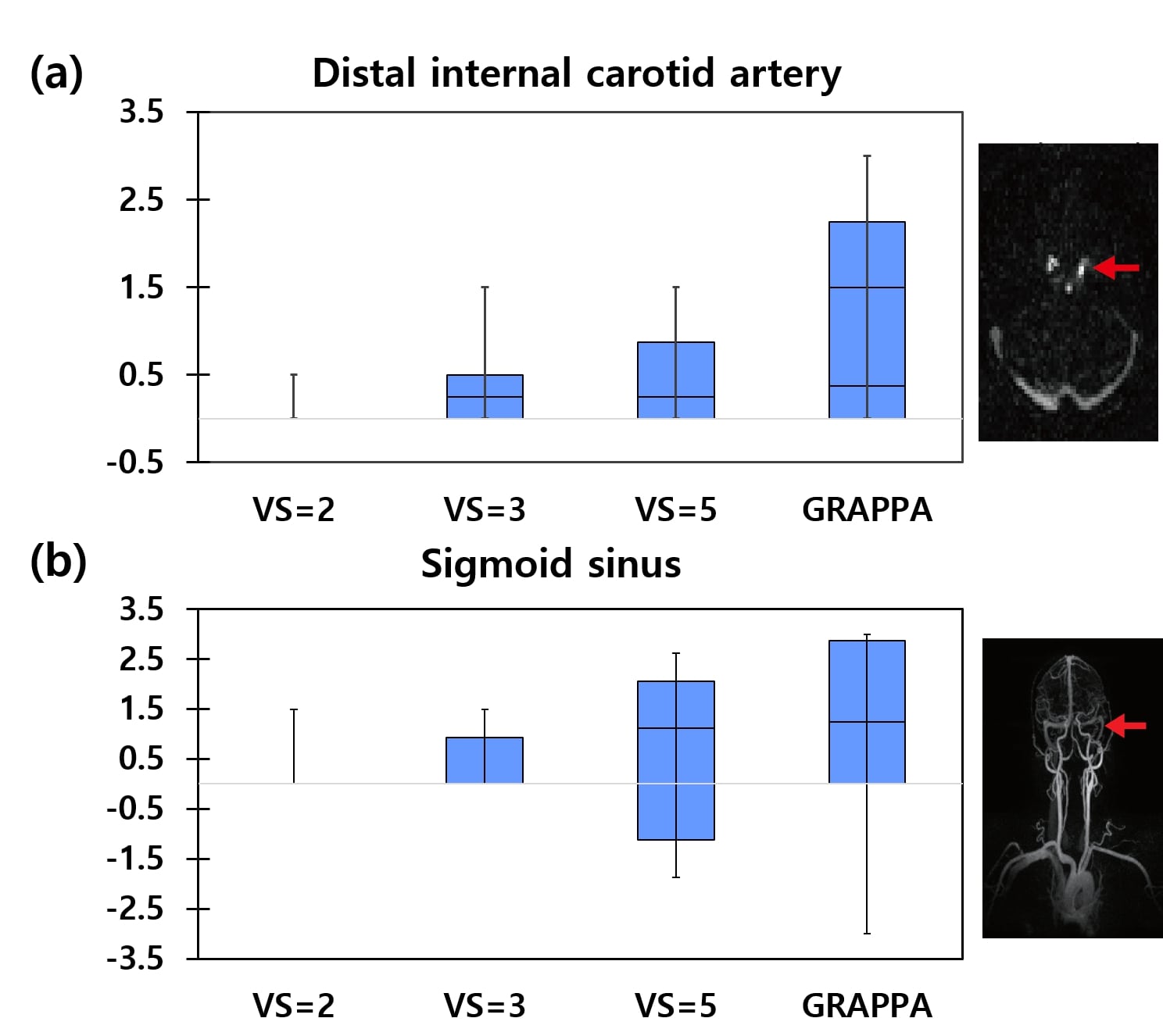}
}
\caption{Box plot of relative start point to peak intensity of distal internal carotid artery and sigmoid sinus with respect to VS. Here, the statistics were calculated using eight patient data sets. {The distal internal carotid artery and the sigmoid sinus are indicated by the red arrows.}
}
\label{fig:temporal_dynamics}
\end{figure}

\begin{figure*}[!hbt] 	
\center{ 
\includegraphics[width=17cm]{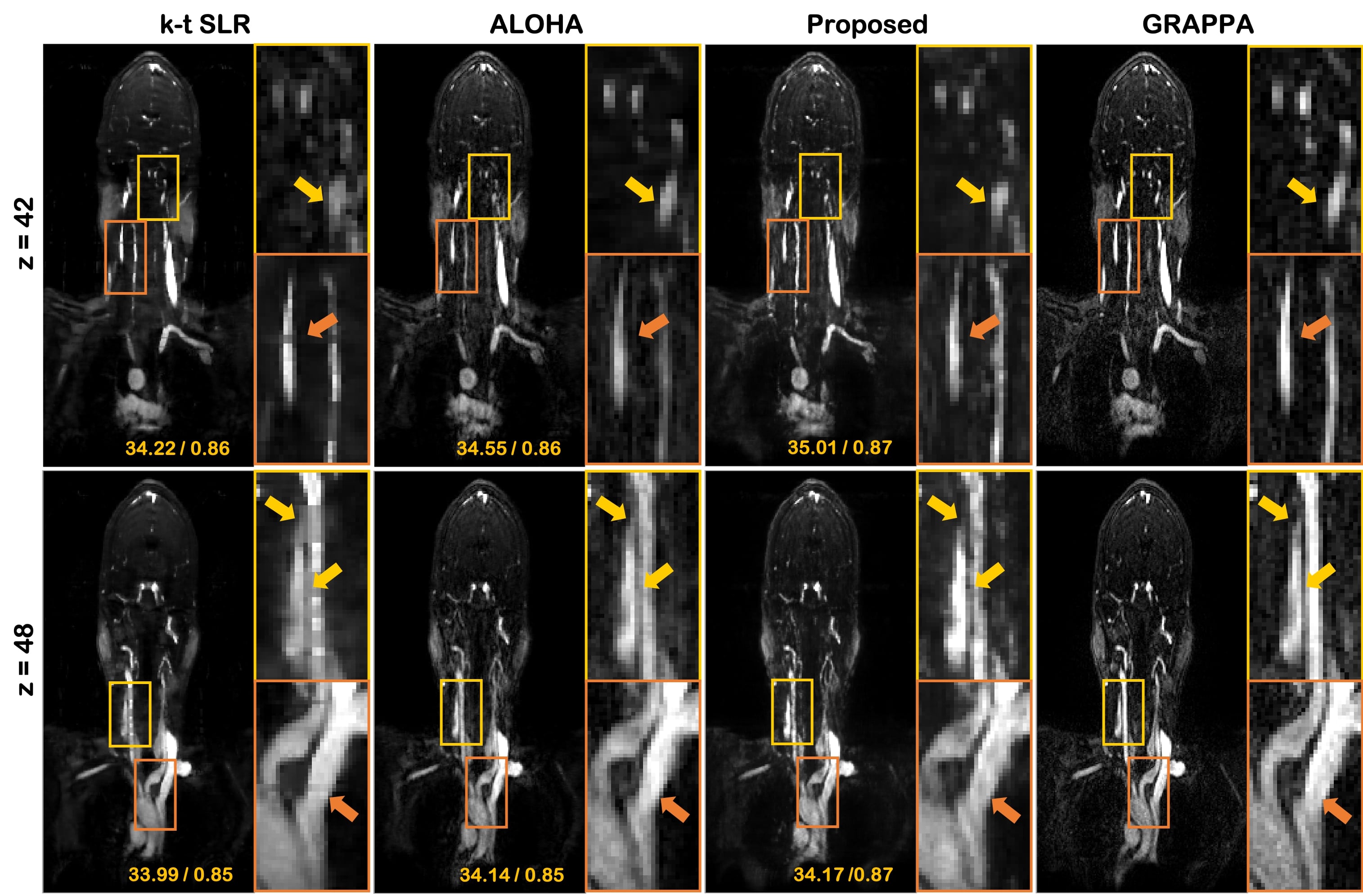}
}
\caption{{Coronal views of reconstruction results of raw data, k-t SLR, ALOHA, GRAPPA and proposed method. 
For the reconstruction using k-t SLR, ALOHA, and the proposed method, sub-sampled data with VS=2 were used. Since GRAPPA can only be applied to the regular sampling pattern, we used the sub-sampled data with VS=5 for GRAPPA reconstruction. The values in the corners are PSNR / SSIM index for individual image with respect to the GRAPPA reconstruction. Due to the lack of ground-truth, the GRAPPA reconstruction
with the most similar images are used as reference in calculating PSNR and SSIM.}}
\label{fig:compare_algorithm}
\end{figure*}

By inspection, we can see that the flow of contrast agent is rather abrupt in the GRAPPA reconstruction. For instance, there was a rapid propagation of contrast agent from the $T=11$ frame to  the $T=12$ frame as shown in Fig. \ref{fig:temporal_resol}. This was due to the sliding window combination of several frames prior to the application of GRAPPA. Therefore, results acquired from GRAPPA reconstruction fail to follow the true temporal dynamics faithfully. Degradation of temporal dynamics gets even severer as number of view-sharing increases.
On the other hand, in the reconstructed images using the proposed method, the flow of the contrast agent was captured to a fine temporal scale with VS=2. Minor temporal blurring can be seen when VS=3 as shown in Fig.~\ref{fig:temporal_resol}. With VS=5, which is equal to the view sharing number used in GRAPPA, the spatio-temporal resolution of the proposed method were nearly identical to the results of GRAPPA.
More specifically, inspecting the raw images carefully we can observe the location of contrast agent flow at each time frame. As the number of view sharing increases, which corresponds to severer temporal blending, flow appear at locations where it should not be seen. For instance, at $ T = 12 $ we can clearly see the temporal resolution degradation with the increase in the number of view sharing. With GRAPPA reconstruction, the detail of the spread of the contrast agent was influenced by the high-frequency region in the $k$-space so that the temporal dynamics of the future frame was erroneously incorporated in the current frame. On the other hand, our method provides high spatial resolution images for various view sharing numbers, although the spatial resolution improves with more view sharing at the cost of temporal blurring. Since  we can reconstruct high quality results even from low view sharing factor, we could investigate various spatio-temporal resolution trade-off.

For quantitative analysis, the difference between the start point to peak intensity of the reconstructed images and that of the raw data was calculated using the eight patient data sets for detailed comparison of temporal dynamics. Distal internal carotid artery and sigmoid sinus were chosen for the box plots of the relative start point to peak intensity as shown in Fig.~\ref{fig:temporal_dynamics}.
As the number of view sharing frames increases, temporal dynamics is degraded. 
In addition, the temporal dynamics of reconstructed images with VS=5 using the proposed network had a similar or better tendency to those using GRAPPA, which implies the enhancement of temporal resolution in the reconstructed results using the proposed unsupervised learning.

\begin{figure*}[!hbt] 	
\center{ 
\includegraphics[width=18cm]{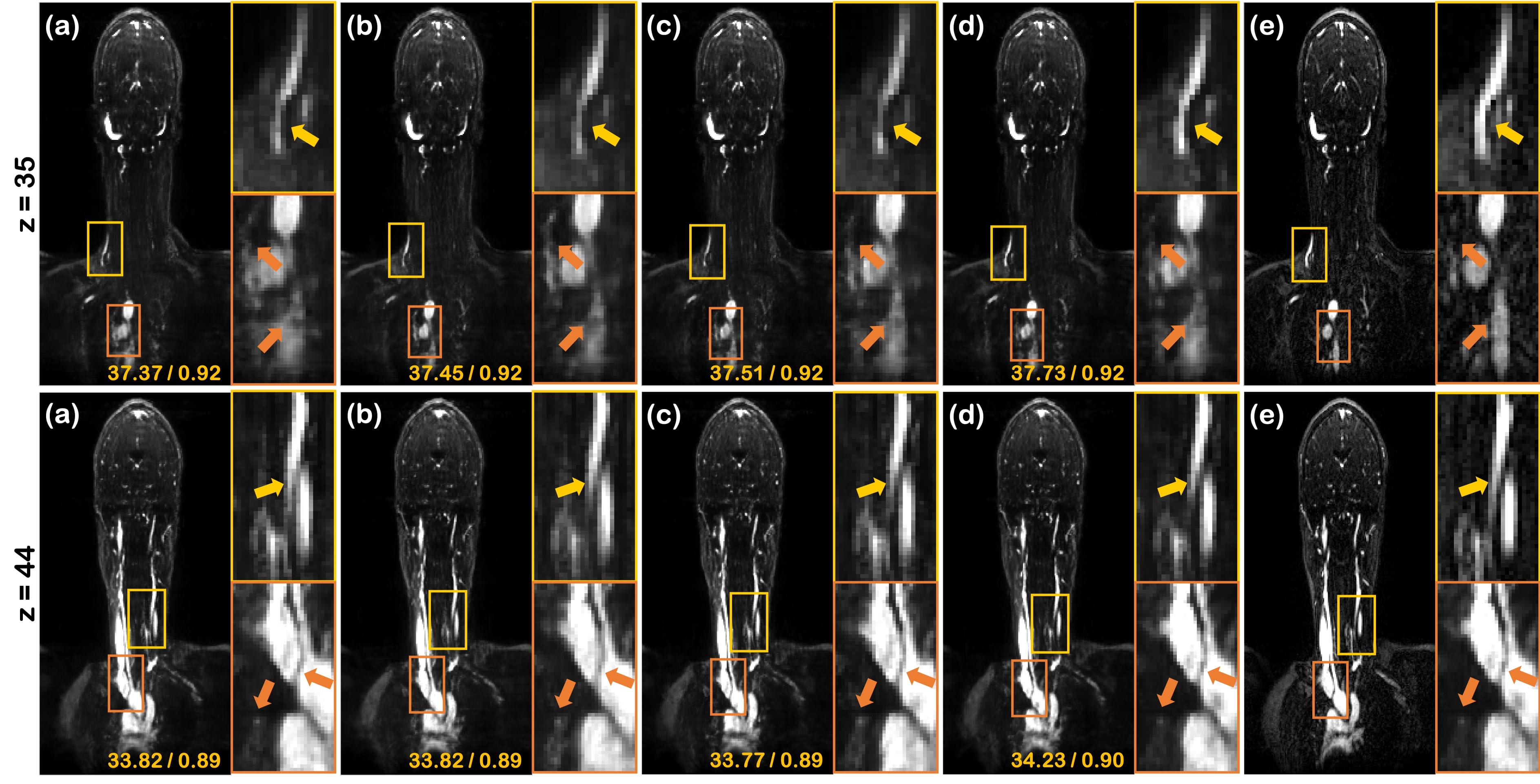}
}
\caption{Coronal view of reconstruction results using an ablated network without (a) $\ell_{freq}$ and $\ell_{identity}$, (b) $\ell_{identity}$,  and (c) $\ell_{freq}$, respectively. Reconstruction results by (d) the proposed method, and (e) GRAPPA, respectively. GRAPPA is reconstructed
with VS=5 and the others are reconstructed with VS=2.  The values in the corners are PSNR / SSIM index for individual image with respect to the GRAPPA reconstruction. Due to the lack of ground-truth, the GRAPPA reconstruction
with the most similar images are used as reference in calculating PSNR and SSIM.}
\label{fig:ablation}
\end{figure*}

\subsection{Comparison with existing algorithms}

We further validate the effectiveness of our unsupervised learning method in contrast to the state-of-the-art CS approaches - ALOHA and k-t SLR. We compared the performance between the CS approaches and the proposed method with VS=2.
As shown in Fig. \ref{fig:compare_algorithm}, reconstructed images with k-t SLR can be characterized with blurriness  so that the details in reconstructed images with k-t SLR cannot be distinguished well. On the other hand,  images reconstructed with ALOHA are sharper, and this sharpness is further improved using the proposed method, which is crucial in clinical situations.

Due to the lack of the ground-truth, exact quantification of the performance was difficult.
Instead,  we calculated the PSNR and SSIM index with respect to the GRAPPA reconstruction with VS=5 that has the most similar
reconstruction. The main assumption is that after the saturation of the contrast agent, the differences in the image may be small. 
 As shown in Fig.~\ref{fig:compare_algorithm}, in terms of quantitative metrics - PSNR and SSIM index, our unsupervised learning consistently outperforms ALOHA. Specifically, our method outperforms ALOHA by 0.03$\sim$0.46 dB in PSNR and exceeds k-t SLR with a large margin. 
The acceleration factor of the data to which the proposed method is applied is $ R = 13.93$, which is much higher than $ R = 6.90$ used for GRAPPA reconstruction. However, the proposed method offers a spatial resolution comparable to that of the GRAPPA reconstruction.
Furthermore, the computational time for slice using the proposed method was about 0.01 sec, while that the computational times for
 ALOHA, k-t SLR were 84.61, and 217.57 sec, respectively. This confirmed that the proposed method is computationally more efficient than ALOHA and k-t SLR.

\section{Discussion}
\subsection{Ablation study}

To analyze the role of individual losses in the proposed method, ablation studies were performed, by excluding the frequency loss $\ell_{freq}$ and/or the identity loss $\ell_{identity}$ under the same training conditions. The reconstruction results at $R=13.93$ acceleration are illustrated in Fig. \ref{fig:ablation}.  As we do not have ground-truth image, 
{the most similar reconstruction images using GRAPPA with VS=5} are illustrated in in the last column.

 The results provided by the network trained without $\ell_{freq}$ and $\ell_{identity}$ are demonstrated in the first column. In the second column and third column, the results excluding $\ell_{identity}$ and $\ell_{freq}$, respectively, are illustrated. The reconstructed images using the proposed method are shown in the fourth column. 
The networks trained without $\ell_{freq}$ and $\ell_{identity}$ (Fig. \ref{fig:ablation}(a)), and $\ell_{freq}$ (Fig. \ref{fig:ablation}(c)) provided more blurry results than the proposed method (Fig. \ref{fig:ablation} (d)). This confirmed that enforcing the frequency domain identity loss
to maintain the acquired $k$-space samples can help to reconstruct the details.
 By comparing Fig. \ref{fig:ablation}(b) and Fig. \ref{fig:ablation}(d), it was shown that the region with contrast dynamics can be reconstructed well without any deformation by the proposed method.
 In terms of PSNR and SSIM, the proposed network persistently outperformed the ablated  networks.
 More specifically, our method is about $0.4\sim 1.2$dB better, which demonstrated the importance of $\ell_{identity}$ and $\ell_{freq}$ for our unsupervised learning.

\begin{figure}[!hbt] 	
\center{ 
\includegraphics[width=8cm]{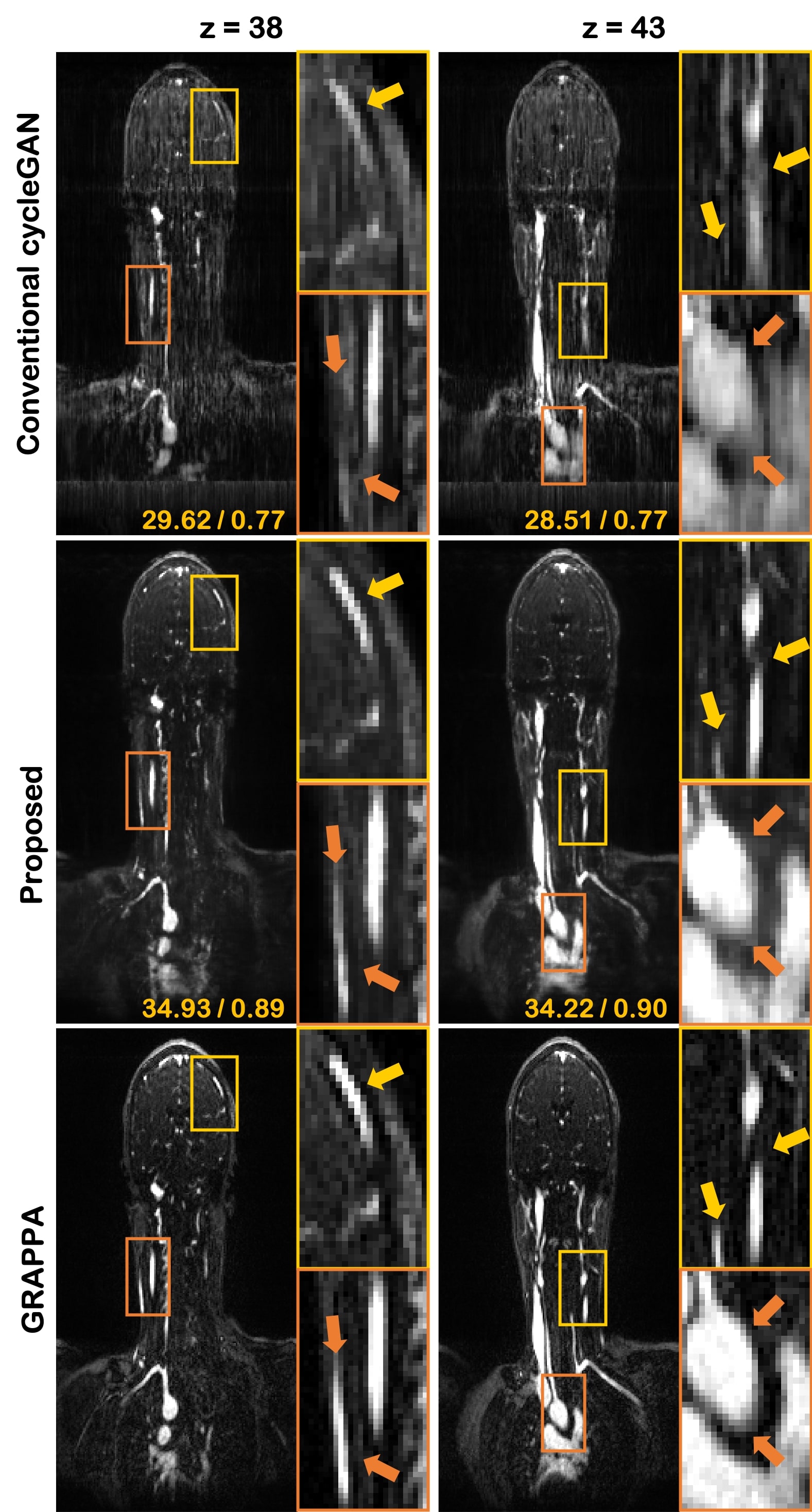}
}
\caption{{Coronal views of reconstruction results using conventional cycleGAN, GRAPPA, and the proposed method. For the reconstructions using conventional cycleGAN and the proposed method, sub-sampled data with VS=2 were used. Since GRAPPA can only be applied to the regular sampling pattern, we used the undersampled data with VS=5 for GRAPPA reconstruction. The values in the corners are PSNR / SSIM index for individual image with respect to the GRAPPA reconstruction. Due to the lack of ground-truth, the GRAPPA reconstruction
with the most similar images are used as reference in calculating PSNR and SSIM.}}
\label{fig:compare_cycleGAN}
\end{figure}

\subsection{Comparison with conventional cycleGAN}

Fig. \ref{fig:compare_cycleGAN} shows the results by the conventional cycleGAN \cite{isola2017image} and the proposed method using
the data with VS=2. Again, due to the
lack of the ground-truth data, we illustrate the GRAPPA reconstruction with VS=5 that shows the most similar results.  Our optimal transport driven cycleGAN consists of a single pair of generator and discriminator, while the standard cycleGAN consists of two pairs of generator and discriminator. As shown in Fig. \ref{fig:compare_cycleGAN}, the performance of the proposed cycleGAN was significantly better than that of the conventional cycleGAN.
More specifically, the intensity of the vessels in the reconstruction using standard cycleGAN is much weaker compared to the proposed method and GRAPPA. In addition, the details of the vascular structures, which are crucial for an accurate diagnosis, are not well recognizable in the reconstructions with the standard cycleGAN, while are clearly visible in the reconstructions with the proposed cycleGAN.
Since our dataset do not contain any ground-truth, PSNR and SSIM index were calculated with respect to the GRAPPA reconstruction with VS=5 for quantitative comparison between the conventional cycleGAN and the proposed cycleGAN.  The proposed method is about 5.3$\sim$5.7 dB better  compared to the conventional cycleGAN in terms of PSNR.
This result is also consistent with SSIM index  improvement.

This improvement is due to  more stable training in our network architectures, since we only have
 a single pair of  generator and discriminator, whereas in the standard cycleGAN two pairs are necessary which makes the training more difficult.

\section{Conclusion}\label{sec:conclusion}

In this paper, we proposed a novel unsupervised learning method  to improve the temporal resolution of tMRA imaging
and generate reconstruction results at  diverse numbers of view sharing.
In particular, we proposed a novel cycleGAN architecture which requires only a single pair of generator and discriminator by exploiting the deterministic undersampling operation. 
We performed various simulation and in vivo experiments to verify the efficacy of the proposed method for TWIST imaging. Using the proposed unsupervised learning, it was demonstrated that the undersampled images with reduced view sharing can be properly reconstructed, resulting in the improvement of temporal resolution of TWIST imaging. 
Our network can provide various reconstruction results by easily altering the number of view sharing at the inference stage.
Moreover, the proposed method can be used with the existing TWIST acquisition protocol without any modification of the pulse sequence, 
despite the significantly small  computational complexity. Since the matching pairs of downsampled and fully sampled data were not required in our method,  our method may provide an important new research direction that can significantly extend the clinical applications of tMRA.

\appendix
\section*{Proof of Theorem~\ref{thm:dual}}

The proof is a direct extension of the proof in \cite{sim2019optimal}, but we include the following for self-containment.

Using the cost $ c(\Xb,\Yb;\Theta)$ given by Eqs.~\eqref{eq:c1}-\eqref{eq:c4},
the primal problem \eqref{eq:primal} becomes
\begin{align}
\Kd(\Theta)=&\min_\pi \int_{\Xc\times \Yc} c(\Xb,\Yb;\Theta)d\pi(\Xb, \Yb) \\
=& \int_{\Xc\times \Yc}  c_{XY}(\Xb,\Yb) d\pi^*(\Xb, \Yb)  \\
&+ \int_{\Xc}  c_X\left(\Xb \right)  d\mu^*(\Xb) \label{eq:Kdorg}
\end{align}
where $\pi^*$ and $\mu^*$ denote the optimal joint measure and marginal. Furthermore,
\begin{align*}
 c_{XY}(\Xb,\Yb) = & d_{I}\left(\Yb,\mathcal{F}^{-1} \Pc_{\Lambda} \mathcal{F} \Xb\right) +d_{I} \left(\Xb,  G_\Theta(\Yb)\right)\\
  c_X\left(\Xb \right) = &\alpha d_{I}\left(\Xb, G_\Theta(\Xb) \right) \\
  &+\beta d_{F}\left(\Pc_\Lambda\Fc \Xb, \Pc_\Lambda \Fc G_\Theta(\mathcal{F}^{-1} \Pc_{\Lambda} \mathcal{F}\Xb)\right).
\end{align*}
Here, the minimization with respect to the marginal $\mu(\Xb)$ is simple, and the technical difficulty
lies in the minimization with respect to the joint measure $\pi(\Xb,\Yb)$.
According to the Kantorovich dual formulation \cite{villani2008optimal}, we have
\begin{align*}
\Kd_{XY}(\Theta):= & \int c_{XY}(\Xb,\Yb) d\pi^*(\Xb, \Yb) \\
= & \frac{1}{2}\left\{ \max_{\varphi}\int_\Xc\varphi(\Xb)d\mu(\Xb)+ \int_\Yc \varphi^c(\Yb) d\nu(\Yb)\right. \\
&+ \left.  \max_{\psi}\int_\Xc\psi^c(\Xb)d\mu(\Xb)+ \int_\Yc \psi(\Yb) d\nu(\Yb) \right\}
\end{align*}
where the so-called c-transforms $ \varphi^c(\Yb)$ and $\psi^c(\Xb)$ are defined by \cite{villani2008optimal}
\begin{align*}
 \varphi^c(\Yb)
&= \inf_\Xb \{ c_{XY}(\Xb, \Yb) -\varphi(\Xb) \}\\
&= \inf_\Xb \{ d_{I}\left(\Yb,\mathcal{F}^{-1} \Pc_{\Lambda} \mathcal{F} \Xb\right)+d_{I} \left(\Xb,  G_\Theta(\Yb)\right)  -\varphi(\Xb) \}\\
 \psi^c(\Xb)
&= \inf_\Yb \{ c_{XY}(\Xb, \Yb) -\psi(\Yb) \}\\
&= \inf_\Yb \{ d_{I}\left(\Yb,\mathcal{F}^{-1} \Pc_{\Lambda} \mathcal{F} \Xb\right)+d_{I} \left(\Xb,  G_\Theta(\Yb)\right)  -\psi(\Yb) \}\
\end{align*}
Now, instead of finding the $\inf_\Xb$, we choose $\Xb=G_\Theta(\Yb)$.
Similarly, instead of finding the $\inf_\Yb$, we choose $\Yb=\Fc^{-1}\Pc_\Lambda \Fc \Xb$. 
This leads to an upper bound: 
\begin{eqnarray*}
\Kd_{XY}(\Theta)&\leq & 
\frac{1}{2}\left(\ell_{cycle}(\Theta)+ \ell_{OT'}(\Theta)\right)
\end{eqnarray*}
where
\begin{align}
\ell_{cycle}(\Theta) =&  \int_\Xc d_{I} \left(\Xb,  G_\Theta(\Fc^{-1}\Pc_\Lambda \Fc \Xb)\right)d\mu(\Xb)  \notag\\
& +\int_\Yc  d_{I}\left(\Yb,\mathcal{F}^{-1} \Pc_{\Lambda} \mathcal{F} G_\Theta(\Yb)\right)d\nu(\Yb) \label{eq:cycle}\\
 \ell_{OT'}(\Theta) = & \max_{\varphi}\int_\Xc \varphi(\Xb)  d\mu(\Xb) - \int_\Yc \varphi(G_\Theta(\Yb))d\nu(\Yb)  \notag\\
 &+ \max_{\psi}\int_{\Yc} \psi(\Yb)  d\nu(\Yb) - \int_\Xc \psi(\mathcal{F}^{-1} \Pc_{\Lambda} \mathcal{F} \Xb)  d\mu(\Xb)
\end{align}
Now, using Kantorovich potentials that satisfy \eqref{eq:bound}, we have
\begin{align*}
 \varphi(\Xb)-&\varphi(G_\Theta(\Yb))\leq d_I(\Xb, G_\Theta(\Yb)) \\
&\leq  d_I(\Xb, G_\Theta(\Yb))+ d_{I}\left(\Yb,\mathcal{F}^{-1} \Pc_{\Lambda} \mathcal{F} \Xb\right) \\
 \psi(\Yb)-&\psi(\mathcal{F}^{-1} \Pc_{\Lambda} \mathcal{F} \Xb)\leq  d_{I}\left(\Yb,\mathcal{F}^{-1} \Pc_{\Lambda} \mathcal{F} \Xb\right) \\
&\leq  d_I(\Xb, G_\Theta(\Yb))+ d_{I}\left(\Yb,\mathcal{F}^{-1} \Pc_{\Lambda} \mathcal{F} \Xb\right) 
\end{align*}
This leads to the following lower-bound
\begin{eqnarray*}
\Kd_{XY}(\Theta)
&\geq& \frac{1}{2}\ell_{OT'}(\Theta)
\end{eqnarray*}
By collecting the two bounds, we have
\begin{eqnarray*}
| \Kd_{XY}(\Theta) -\Dd_{XY}(\Theta) | \leq \frac{1}{4}\ell_{cycle}(\Theta,\Hc).
\end{eqnarray*}
where $\Dd_{XY}(\Theta)$ is defined as
\begin{align*}
\Dd_{XY}(\Theta):=\frac{1}{2}\ell_{OT'}(\Theta)+\frac{1}{4}\ell_{cycle}(\Theta)
\end{align*}
Therefore, the following primal problem of the optimal transport:
\begin{align}
\min\limits_{\Theta}\Kd_{XY}(\Theta) 
\end{align}
 can be equivalently represented by a constrained optimization problem:
\begin{align} \label{eq:dualc}
\min\limits_{\Theta}~&\Dd_{XY}(\Theta) \quad \mbox{subject to}~\ell_{cycle}(\Theta)=0
\end{align}
Using Lagrangian multiplier,  the constrained optimization problem \eqref{eq:dualc} can be converted to an unconstrained
optimization problem: 
\begin{align*} 
\min\limits_\Theta\ell_{OT'}(\Theta)+\gamma\ell_{cycle}(\Theta)
\end{align*}
where $\alpha$ is a Lagrangian parameter and $\gamma=\alpha+\frac{1}{2}$.
By implementing the Kantorovich potential using CNNs with  parameters $\Upsilon$ and $\Xi$,
i.e. $\varphi:=\varphi_\Upsilon$ and $\psi:=\psi_\Xi$, 
we have the the following cycleGAN problem:
\begin{align}\label{eq:cycleGAN}
\min\limits_{\Theta}\max\limits_{\Upsilon,\Xi}\ell(\Theta;\Upsilon,\Xi)
\end{align}
where
\begin{eqnarray}
\ell(\Theta;\Upsilon,\Xi) =  \ell_{GAN}(\Theta;\Upsilon) + \ell_{dual}(\Xi)+ \gamma\ell_{cycle}(\Theta)\notag
\end{eqnarray}
where $ \ell_{cycle}(\Theta)$ denotes the cycle-consistency loss in \eqref{eq:cycle}  and
$\ell_{GAN}(\Theta;\Upsilon)$ is the GAN loss given by:
\begin{align}\label{eq:ellGAN}
\ell_{GAN}(\Theta;\Upsilon)=& \int_{\Xc} \varphi_\Upsilon(\Xb)d\mu(\Xb) - \int_\Yc \varphi_\Upsilon(G_\Theta(\Yb)) d\nu(\Yb)   \notag \\ 
\ell_{dual}(\Xi)=& \int_{\Yc} \psi_\Xi(\Yb)d\nu(\Yb) - \int_\Xc \psi_\Xi(\mathcal{F}^{-1} \Pc_{\Lambda} \mathcal{F} \Xb)d\mu(\Xb)
\end{align}
Note that the parameter $\Xi$ does not affect the generator $G_\Theta(\Yb)$, since the forward operator
$\mathcal{F}^{-1} \Pc_{\Lambda} \mathcal{F} $ is fixed.
Therefore, our simplified min-max optimization problem becomes
\begin{align}\label{eq:cycleGAN2}
\min\limits_{\Theta}\max\limits_{\Upsilon}\ell_{GAN}(\Theta;\Upsilon) + \gamma\ell_{cycle}(\Theta)
\end{align}
Since this is from the dual formulation of $\Kb_{XY}(\Theta)$, the final step of the proof includes
the remaining terms in \eqref{eq:Kdorg}.
This leads to the following min-max problem:
\begin{align}\label{eq:cycleGANfinal}
\min\limits_{\Theta}\max\limits_{\Upsilon}\ell_{GAN}(\Theta;\Upsilon) + \gamma\ell_{cycle}(\Theta) +\alpha\ell_{identity}(\Theta)+\beta\ell_{freq}(\Theta)
\end{align}
where $\ell_{identity}$ and $\ell_{freq}$ are defined by \eqref{eq:identity} and \eqref{eq:freq}, respectively.
This concludes the proof.

\bibliographystyle{IEEEtran}
\bibliography{ref,submit_bib}


\end{document}